%% file: bubble_arxiv.tex
\documentclass[a4paper,journal,twoside]{IEEEtran}
\IEEEoverridecommandlockouts

\newcommand\hmmax{0}
\newcommand\bmmax{0}

\usepackage[ruled,vlined,linesnumbered]{algorithm2e} 
\usepackage{graphicx}
\usepackage{tikz}
\usepackage{pgfplots}
\usepackage{xcolor}
\usepackage{color, colortbl}
\usepackage{amsmath}
\usepackage{bbm}
\usepackage{amsfonts, amssymb}
\usepackage{mathtools}
\usepackage{enumerate}
\usepackage{cite}
\usepackage[nolist]{acronym}
\usepackage{amsmath}

\usepackage{booktabs} 		
\usepackage{diagbox}		
\usepackage{upgreek}
\usepackage{bbm}
\usepackage{Files/bm}
\usepackage{Files/winsnotation}
\usepackage{Files/Symbols_OPL}
\usepackage{Files/LVcolours} 

\newcommand{\rev}[1]{#1}

\usepackage{setspace}	 	
\usepackage{comment}
\usepackage{physics}
\usepackage{enumitem}

\newlist{myitemize}{itemize}{3}
\setlist[myitemize,1]{label=1.,leftmargin=1em}
\setlist[myitemize,2]{label=$\rightarrow$,leftmargin=0.75em}
\setlist[myitemize,3]{label=$\diamond$}

\usepackage{array}
\newcolumntype{C}[1]{>{\centering\arraybackslash}p{#1}}

\makeatletter
\def\endthebibliography{%
  \def\@noitemerr{\@latex@warning{Empty `thebibliography' environment}}%
  \endlist
}
\makeatother
\usepackage{amsthm}
\theoremstyle{definition}

\newtheorem{theorem}{Theorem}

\newtheorem{corollary}{Corollary}
\newtheorem{example}{Example}
\usepackage{algpseudocode,amsmath}

\newcommand{\PauliX}{\M{X}}
\newcommand{\PauliY}{\M{Y}}
\newcommand{\PauliZ}{\M{Z}}

\pgfplotsset{compat=1.17}

\begin{document}

\title{Bubble Clustering Decoder for \\ Quantum Topological Codes }



\author{Diego Forlivesi,~\IEEEmembership{Graduate~Student~Member,~IEEE,}
Lorenzo~Valentini,~\IEEEmembership{Member,~IEEE,}
        and~Marco~Chiani,~\IEEEmembership{Fellow,~IEEE}
\thanks{The authors are with the Department of Electrical, Electronic, and Information Engineering ``Guglielmo Marconi'' and CNIT/WiLab, University of Bologna, V.le Risorgimento 2, 40136 Bologna, Italy. E-mail: \{diego.forlivesi2, lorenzo.valentini13, marco.chiani\}@unibo.it. 
This work is supported in part by PNRR MUR project PE0000023-NQSTI (Italy) financed by the European Union – Next Generation EU.
}
}

\maketitle 
\markboth{}{Forlivesi, Valentini, Chiani: Bubble Clustering Decoder for Quantum Topological Codes}

\input{Files/Acronimi_SICMMA.tex}
\setcounter{page}{1}

\begin{abstract}
Quantum computers are highly vulnerable to noise, necessitating the use of error-correcting codes to protect stored data. 
Errors must be continuously corrected over time to counteract decoherence using appropriate decoders. 
Therefore, fast decoding strategies capable of handling real-time syndrome extraction are crucial for achieving fault-tolerant quantum computing.
\rev{In this paper, we introduce the bubble clustering (BC) decoder for quantum surface codes, which serves as a low-latency replacement for MWPM, achieving significantly faster execution at the cost of a slight performance degradation.} 
This speed boost is obtained leveraging an efficient cluster generation based on bubbles centered on defects, and avoiding the computational overhead associated with cluster growth and merging phases, commonly adopted in traditional decoders. 
Our complexity analysis reveals that the proposed decoder operates with a complexity on the order of the square of the number of defects. 
For moderate physical error rates, this is equivalent to linear complexity in the number of data qubits.

\end{abstract}

\begin{IEEEkeywords} Quantum Error Correcting Codes, Quantum Communications, Quantum Computing, Surface Codes
\end{IEEEkeywords}

\section{Introduction}

The groundbreaking potential of quantum computing is fundamentally constrained by the fragility of quantum states, which are highly vulnerable to decoherence~\cite{Pre:18, Pfi:23, Sho:95}.
To address this challenge, \ac{QEC} has emerged as a critical area of research to enable quantum computing and communication~\cite{NieChu:10, Bab:19 ,CacCalVan:20, ZorDePGio:23}. 
In this context, numerous short \acp{QECC} have been proposed in the literature as early methods for encoding quantum information~\cite{Sho:95, Got:96, Ste:96, Laf:96, Got:09, Bom06:colorCodes, ChiVal:20a}.

In recent years, significant efforts have been made to find quantum codes that are easier to implement.
Among them, one of the most promising approaches in quantum error correction is represented by surface codes~\cite{BraKit:98, FowMarMar:12}. 
These topological surface codes encode logical qubits onto a two-dimensional lattice of physical qubits, offering robust error protection with relatively practical implementation requirements~\cite{Rof:19, Vui19:GaugeFixing}.
With their high error threshold, locality, and scalability, surface codes have therefore emerged as a leading candidate for implementation in fault-tolerant quantum computing architectures~\cite{ZhaYouYe:22, AchRajAle:22, BluDolEve:23}.

The \acf{MWPM} decoder is currently the prevalent choice for surface code decoding \cite{Higg:22, Bro:23}. 
While it offers high threshold error rates, its high order polynomial time complexity can lead to latency issues, potentially hindering the performance of quantum computation architectures \cite{kol:09, Ter:15}.
Indeed, if the decoder is unable to process measurements quickly enough, the accumulating backlog of syndrome information can lead to an exponential increase in computation time~\cite{Ach:24}.
To overcome this latency problem, the PyMatching sparse blossom implementation was introduced, delivering a speed boost of up to 100 times compared to the standard \ac{MWPM} decoder~\cite{HigGid:23}. 
This improvement was achieved by focusing exclusively on the essential edges of the lattice needed for the blossom algorithm and by incorporating techniques such as compressed tracking.
A specialized version of the Sparse Blossom algorithm was optimized for the decoding of surface codes with distance $d = 5$ and $d = 7$ \cite{Ach:24}. 
It operates by employing multiple threads \rev{to} process syndrome information from different spacetime regions, which are combined to find a global minimum-weight perfect matching. 
\rev{Hereby, a greedy edge reweighting strategy  accounts for Y-type error correlations and improves accuracy.}

Among the alternatives to the \ac{MWPM} discussed in the literature, the \ac{UF} decoder stands out as a prominent choice \cite{Del:21, Del:20}. 
This decoder achieves an almost-linear worst-case runtime relative to the number of physical qubits.
However, despite its ability to correct errors up to the code distance, the \ac{UF} decoder is less accurate than the MWPM decoder \cite{Wu:22, Hig:23}.
The most time-consuming step of this decoding procedure is the syndrome validation phase. 
Specifically, the \ac{UF} decoding begins by initializing clusters, each containing a single vertex representing an ancilla that has detected an error during the syndrome measurement, often referred to as defect. 
Finally, each cluster is processed in order to obtain a suitable matching.

Alternative approaches are provided by the \ac{STM} and \ac{RFire} decoders \cite{ForValChi:24STM}.
\rev{Specifically, we define the defect graph as a complete graph in which the vertices represent defects, and the edges are weighted based on the error probabilities of the qubits that lie between each pair of defects.
In the case of \ac{i.i.d.} data qubit errors, these weights are equivalent to the number of qubits separating the pair.} 
It has been proven that, under a particular metric, all distinct matchings in the defect graph of a surface code are equivalent to the minimum-weight matching. Consequently, using this metric to identify errors ensures accurate correction up to the code minimum distance.
Building on this, the algorithms compute a minimum spanning tree over the defect graph to identify an appropriate matching for localizing the errors.
These decoders offer a notable improvement in decoding time compared to the standard LEMON implementation of the \ac{MWPM} decoder, though it comes with some trade-off in performance.

Finally, several promising developments in neural network-based decoders for surface codes have been reported. 
Specifically, an artificial neural network decoder designed for large surface code distances is introduced in \cite{Gic:23}. 
A key advantage of this approach is its near-constant execution time as the code distance increases.
Additionally, a neural network decoder for quantum surface codes, scalable to tens of thousands of qubits, has been proposed under depolarizing noise \cite{mei:22}. 
This method demonstrates improved error thresholds for depolarizing noise across various physical error rates, outperforming the standard union-find decoder.
\rev{Moreover, in \cite{miao:25}, the authors propose a \ac{BP} decoder for toric codes using overcomplete check matrices, and extend the neural \ac{BP} decoder from suboptimal binary to quaternary \ac{BP} decoding.}
As evidenced by the numerous proposals in the field, the pursuit of a fast decoder suitable for real-time decoding of surface codes remains a highly active research topic~\cite{HigGid:23, Del:21, Del:20, Wu:22, Hig:23, ForValChi:24STM, Pac24:QBitFlipping, Bra14:MPS, Rof20:BP, iOl23:SurveyQDec, Bar:25}.

In this paper, we introduce the \ac{BC} decoder for quantum topological codes.
In particular, our decoder can be used to decode any topological code that is typically decoded using \ac{MWPM}-based decoders.
\rev{For the sake of clarity and exposition, we focus on surface codes.
Among topological codes these are obtained using simple classical repetition codes with full-rank parity check matrices~\cite{ValForChi25:CylMob}.
Differently from toric codes, they involve boundaries which necessitate the additional consideration of ghost ancillas when designing a decoder. 
For this reason and the practical relevance of these codes, both the examples and the explanation of the decoding algorithm in this paper are centered on surface codes.} 
Using as a foundation the asymptotic performance analysis of quantum codes~\cite{ForValChi24:MacW}, we have tailored the decoding strategy to preserve the error correction capability of a code, while addressing the most harmful error patterns with weight beyond its correction capability.
The decoder is able to rapidly generate a series of clusters avoiding the time-consuming tasks of cluster growth and merging.
We begin by placing each defect at the center of a bubble with a uniform radius, carefully determining the value of the radius to ensure error correction up to the code distance. 
Next, we efficiently group defects into cluster trees, storing only the edges between adjacent defects.
Finally, we peel these trees and obtain a suitable matching.
To validate our decoder, we compare its time consumption and performance with several state-of-the-art decoders by extensive simulations.


This paper is organized as follows. Section~\ref{sec:preliminary} introduces preliminary concepts about \ac{QECC} along with a detailed discussion of efficient decoders for surface codes. 
In Section~\ref{sec:STD}, we thoroughly describe the proposed \ac{BC} decoder. 
Section~\ref{sec:complexity} provides a theoretical proof about the error correction capability preservation of our decoder and its complexity analysis.
Finally, numerical results are presented in Section~\ref{sec:NumRes}.

\section{Preliminaries and Background}
\label{sec:preliminary}

\subsection{Stabilizer Formalism}
\label{subsec:QEC}

The Pauli operators are denoted by $\PauliX, \PauliY$, and $\PauliZ$. 
A \ac{QECC} encoding $k$ logical qubits $\ket{\varphi}$ into a codeword of $n$ data qubits $\ket{\psi}$, with minimum distance $d$, is represented as $[[n,k,d]]$. 
This code can correct all error patterns involving up to $t = \lfloor(d-1)/2 \rfloor$ data qubits.
In the stabilizer formalism, each code is defined by $n-k$ independent and commuting operators $\M{G}_i \in \mathcal{G}_n$, known as stabilizer generators or simply generators, where $\mathcal{G}_n$ is the Pauli group on $n$ qubits  \cite{Got:09, NieChu:10}. 
The subgroup of $\mathcal{G}_n$ generated by all combinations of the $\M{G}_i$ is called the stabilizer and denoted as $\mathcal{S}$. The code $\mathcal{C}$ consists of quantum states $\ket{\psi}$ stabilized by $\mathcal{S}$, meaning they satisfy $\M{S}\ket{\psi}=\ket{\psi}$ for all $\M{S} \in \mathcal{S}$, or equivalently, $\M{G}_i \ket{\psi}=\ket{\psi}$ for $i=1, 2, \ldots, n-k$. 
Operators that commute with the stabilizer group but are not part of it are called logical operators.
The stabilizer generators specify measurements on quantum codewords that do not alter the original quantum state, and these measurements are conducted using additional ancilla qubits.
When an error $\M{E} \in \mathcal{G}_n$ affects a codeword, transforming the state to $\M{E}\ket{\psi}$, it is possible to extract a binary sequence $\V{s}$ (the error syndrome). 
The $i$-th entry $s_i$ of this sequence is zero if $\M{G}_i$ commutes with $\M{E}$ and one if $\M{G}_i$ anticommutes with it. 
In particular, ancillas measuring $s_i = 1$ are often called defects.
This enables quantum error correction through error syndrome decoding, using the binary sequence $\V{s}$ as input. 
\ac{CSS} \rev{codes} are an important class of quantum stabilizer codes \cite{CalSho:96,Ste:96}.
By definition, these codes have some of their generators composed of only Pauli $\M{X}$ operators and the others made up of only Pauli $\M{Z}$ operators. 
Consequently, \ac{CSS} codes facilitate efficient decoding by allowing $\M{X}$ and $\M{Z}$ errors to be corrected independently.

Surface codes represent a significant category of \ac{CSS} stabilizer codes, notable for arranging qubits on a planar sheet \cite{BraKit:98, DenKitLan:02, HorFowDev:12, AtaTucBar:21}. 
This configuration necessitates only nearest-neighbor interactions between qubits and enables a single round of stabilizer measurements through parallel operations \cite{BluDolEve:23}. 
In surface codes, logical operators are defined based on the topology of the lattice and the paths they form\cite{DenKitLan:02, FowSteGro:09, ForValChi:23}.
In particular, they can be visualized on the lattice: $\M{Z}_L$ operators are represented by paths extending horizontally from a boundary to the other one. Similarly, $\M{X}_L$ operators are represented by paths extending vertically.
The logical operators \rev{are} the undetectable error patterns of a quantum error correcting code, and for this reason they are crucial in understanding the performance of a code and in the design of a decoder~\cite{ForValChi:23, ForValChi:24STM}.

If we consider a depolarizing channel, where the different Pauli errors occur with the same probability, the logical error rate of a $t$-error correcting code can be approximated as \cite{ForValChi24:MacW}
\begin{align}
\label{eq:error_probWithBetaApprox}
p_\mathrm{L} 
&\approx \left(1-\beta_{t+1}\right) \binom{n}{t+1}p^{t+1} 
\end{align}
where $\beta_{t+1}$ is the fraction of errors of weight $t+1$ that the decoder is able to correct, and the approximation is valid for sufficiently low physical error rate $p$ and assuming that the decoder is able to correct all errors of weight up to $t$.

\subsection{Efficient decoders for surface codes }
\label{subsec:MWPM}

A matching in a graph is defined as a set of edges where no two edges share a common vertex. 
A perfect matching is a type of matching that includes every vertex in the graph \cite{Die:08}. 
A minimum weight perfect matching is a perfect matching with the smallest possible total edge weight.
In the context of quantum error correction, the \ac{MWPM} decoder builds a graph where vertices represent defects.
The edges in this graph are weighted according to the error probabilities of the qubits that lie between each pair of defects. 
Given an error syndrome, the standard implementation of the \ac{MWPM} decoder involves three sequential steps \cite{Higg:22}.
First, Dijkstra's algorithm is used to assign weights and construct the graph of defects as previously described. 
Next, the Blossom algorithm is employed to find the minimum weight perfect matching solution \cite{kol:09}.
Finally, Dijkstra's algorithm is utilized again to map the paired defects back to chains of faulty qubits in the actual lattice.
Note that, when dealing with \ac{i.i.d.} data qubit errors, the process can be significantly simplified by using the Manhattan distance\footnote{The Manhattan distance between two points in an $n$-dimensional space, with coordinates $(x_1, \dots, x_n)$ and $(y_1, \dots, y_n)$, is given by $\sum_{i=1}^{n} | x_i - y_i |$.} between defects. 
Moreover, for surface codes with boundaries, it is essential to introduce ghost defects before applying Dijkstra's algorithm \cite{Had:08}. Indeed, error chains can terminate at a boundary, creating an odd number of defects. 
To address this, the decoder includes a corresponding ghost defect for each real defect. 
In the final graph, to ensure the correct decoding procedure, these ghost defects are connected to each other with zero distance.
A similar implementation of the \ac{MWPM} decoder exhibits a worst-case complexity in the number of nodes $N$ in the graph of $O(N^3 \, \log (N))$, yet empirically, the expected running time for typical instances is approximately $O(N^2)$ \cite{Hig:23,Higg:22}.

Given that executing the three aforementioned steps sequentially can be computationally demanding, recent proposals have emerged to circumvent the Dijkstra step in constructing the edges of the defect graph. 
Notably, the PyMatching sparse blossom implementation dynamically discovers and stores an edge only when necessary for the blossom algorithm's operation \cite{HigGid:23}. 
This proposal adopts an error model where each error mechanism is defined by the generators and logical operators it affects, rather than by its Pauli type and circuit location. 
This approach enables the use of techniques such as compressed tracking, a sparse representation of paths in the defect graph. Compressed tracking stores only the endpoints of a path, along with the logical observables it affects (represented as a bitmask).
Specifically, the authors define compressed edges,  which denote a path through the defect graph linking two defects or a defect and a boundary. These compressed edges retain information about the two endpoints and the list of logical operators flipped by inverting each edge of the real lattice included in the path.
This method enables a highly efficient transition from the final matching to the decoded channel error, resulting from the algorithm. 
Such an implementation of the MWPM is expected to run between $100$ and $1000$ times faster than the standard implementation \cite{Hig:23}.  

Among the various alternatives to the \ac{MWPM} proposed in the literature, one of the most notable is the \ac{UF} decoder \cite{Del:21, Del:20}. 
The \ac{UF} decoder comprises a syndrome validation step and an erasure decoder. First, the process involves transforming the set of Pauli errors into clusters distinguished by an even parity of defects.
During this stage, all defects are initialized as separate clusters. 
Then, in each iteration, every odd cluster extends by half an edge in each direction, facilitating connections between defects rather than between a defect and a boundary. If a cluster encounters another one, they merge, with each defect from the smaller cluster becoming part of the larger one. 
In particular, the use of a tree representation for each cluster enables efficient execution of this step. 
When two odd clusters merge, or when a cluster encounters a boundary, they acquire even parity and cease to grow.
Finally, after a cycle detection and removal within clusters, the erasure decoder employs a peeling decoder to identify an appropriate matching for each cluster.
This decoder shows an almost-linear worst-case running time concerning the number of physical qubits \cite{Del:21}.
However, the \ac{UF} decoder, while offering error correction capability up to the code distance, is less accurate than the MWPM decoder \cite{Wu:22, Hig:23}. 

The \ac{STM} and \ac{RFire} decoders offer two efficient options for decoding surface codes \cite{ForValChi:24STM}. 
The authors observed that, given the defect graph, two distinct matchings are possible, except for the application of some stabilizers: one always corrects the error, while the other introduces a logical operator.
They employ a decision technique that minimizes the number of columns, i.e., horizontally aligned edges in the lattice, traversed an odd number of times.
It has been proven that using this decision technique, all distinct matchings in the defect graph of a surface code are equivalent to the minimum-weight matching. 
Thus, an error with weight $\leq t$ is always corrected. 
Specifically, the \ac{STM} decoder computes a minimum spanning tree of the complete defect graph to obtain a suitable matching.
This can be achieved with a complexity of $O(M \log n_\mathrm{d})$, where $n_\mathrm{d}$ is the number of defects and $M = \binom{n_\mathrm{d}}{2}$ represents the edges in the graph \cite{Cor:22}.
The \ac{RFire} decoder, on the other hand, starts from the complete defect graph and iteratively connects pairs of nearest ancillas to find an appropriate matching. 
These decoders significantly improve decoding time compared to the standard LEMON implementation of the \ac{MWPM} decoder, although they come with some trade-offs in performance.

\subsection{\rev{Comparison with Union-Find Decoder}}
The \ac{BC} decoder fundamentally diverges from the \ac{UF} decoder in several critical aspects.
First, while the \ac{UF} decoder processes syndrome information directly on the real lattice of the surface code, the bubble clustering decoder operates on the defect graph.
As a consequence, the \ac{UF} decoder relies on iterative cluster growth and merging steps.
In contrast, the \ac{BC} decoder eliminates these steps by predefining a maximum cluster radius based on the number of defects, ensuring that any two defects within this radius are grouped into the same cluster \rev{(see Section~\ref{sec:SCP})}.
Moverover, the \ac{UF} decoder requires cycle detection and removal within clusters due to the iterative growth on the lattice, while \ac{BC} avoids cycles inherently by working directly on the defect graph, checking membership before adding defects to clusters.
Note that each cluster in the \ac{UF} decoder may contain not only defects but also a large number of additional vertices corresponding to switched-off generators, along with a significant number of edges. In contrast, each cluster in the \ac{BC} decoder consists exclusively of defects as vertices and includes only the edges directly incident to them.
Consequently, the described peeling procedure is inherently more efficient than the erasure decoding process used in the \ac{UF} algorithm \rev{(see Section~\ref{sec:PP})}.
At this stage, the \ac{UF} decoding terminates, guaranteeing code distance correction. 
In contrast, the \ac{BC} decoder may need to select the final matching based on \eqref{eq:metricCol} to ensure the preservation of the code distance, before mapping each matching to the corresponding faulty qubits in the lattice \rev{(see Section~\ref{subsec:EC})}.

\section{Bubble Clustering Decoder: Description}
\label{sec:STD}

In \cite{ForValChi:24STM} it was proven that both the \ac{STM} and the \ac{RFire} decoders are able to guarantee the error correction capability $t$ of surface codes.
However, due to some uncorrected error patterns of weight $t+1$, a gap in performance between these decoders and the \ac{MWPM} arises.
In particular, we observe that several error patterns, among the ones causing the performance discrepancy, have errors located far away to each other in the lattice grid.
For instance, the error patterns shown in Fig.~\ref{Fig:errorsNONc} cannot be decoded by the \ac{STM} or \ac{RFire} decoders, whereas the \ac{MWPM} decoder successfully decodes them.
To address this issue, we clustered the defects into smaller subsets using an ad-hoc radius, enabling the decoding of multiple erroneous patterns and approaching the performance of \ac{MWPM}. 
Moreover, to further reduce execution times, we developed an algorithm that, while performing the clustering, simultaneously generates a series of trees, i.e., connected graphs without cycles.
In the following Sections, we delve into the details of our \ac{BC} decoder that starts from the error syndrome (i.e., the defect locations) and returns an estimated error pattern. 
In particular, the four phases composing our solution are: radius evaluation phase, a bubble clustering phase, a peeling phase, and an error correction phase (see Fig.~\ref{Fig:blockD}).  

\begin{figure}[t]
 	\centering
 	\includegraphics[width = \columnwidth]{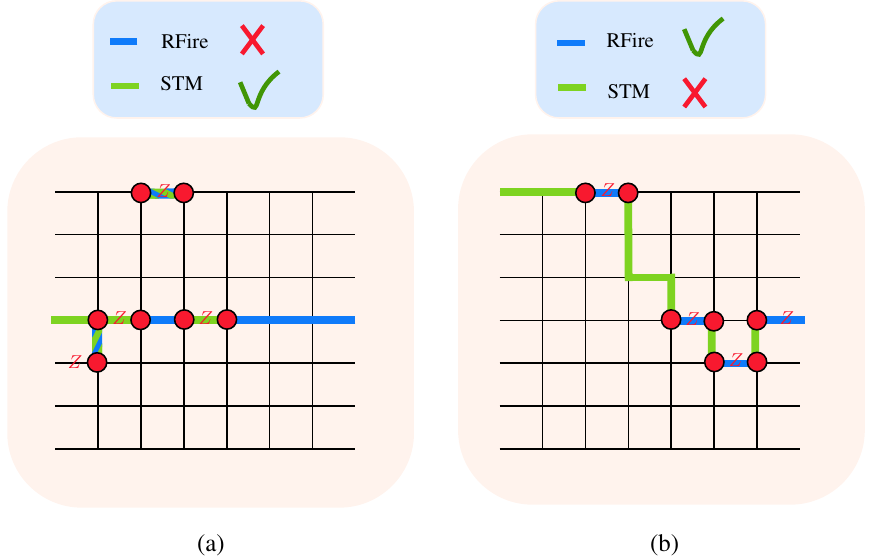}
 	\caption{Examples of non corrected error patterns in a $[[85,1,7]]$ surface code for the \ac{RFire} and \ac{STM} decoders. 
    Each faulty qubit is represented by a red $\M{Z}$ symbol, while defects are illustrated as red dots. 
    The correction operators applied by the RFire decoder are represented by thick blue edges, whereas those applied by the STM decoder are represented by green edges.
    a) An error pattern of weight $t + 1$ that is not corrected by the \ac{RFire} decoder but is corrected by the \ac{STM} decoder.
    b) An error pattern of weight $t + 1$ that is not corrected by the \ac{STM} decoder but is corrected by the \ac{RFire} decoder. Both error patterns are corrected by the \ac{BC} decoder.
  }
 	\label{Fig:errorsNONc}
\end{figure}
\begin{figure*}[t]
 	\centering
     \includegraphics[width = \textwidth]{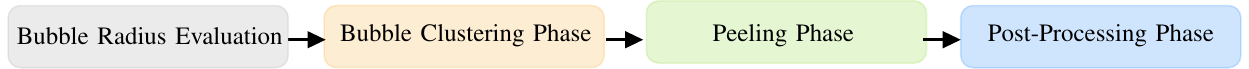}
 	\caption{Block diagram of the bubble clustering decoder.}
 	\label{Fig:blockD}
\end{figure*}
\begin{figure}[t]
 	\centering
\includegraphics[width=0.99\columnwidth]{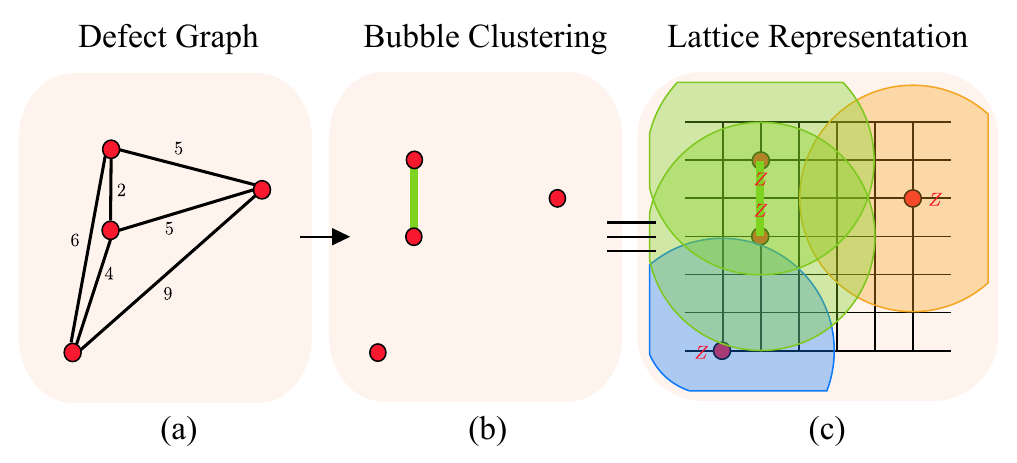}
 	\caption{ \rev{Example of bubble clustering phase for a $[[85,1,7]]$ surface code.}
    }
 	\label{Fig:sphereClustering2}
\end{figure}
\begin{figure*}[t]
 	\centering
     \includegraphics[width = \textwidth]{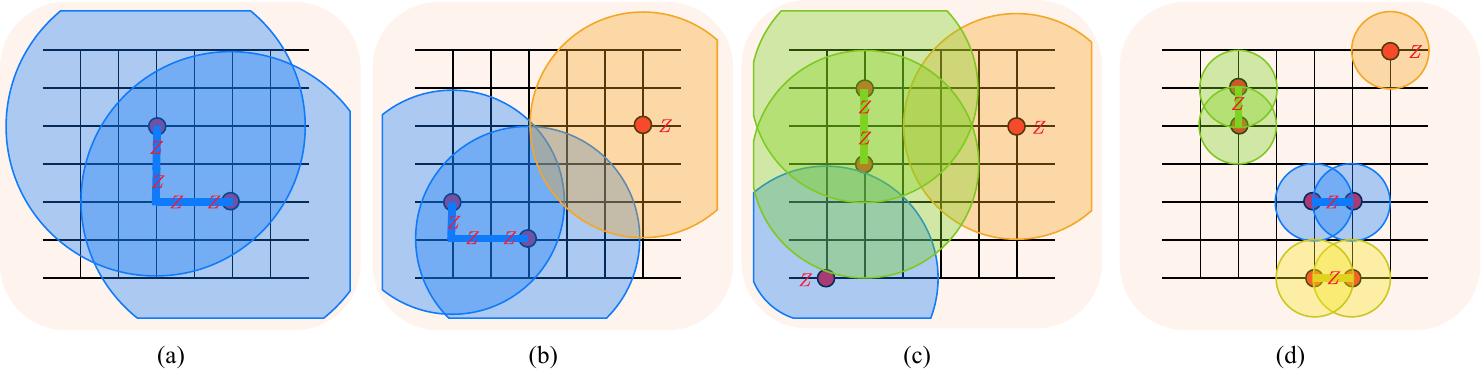}
 	\caption{ Example of bubble clustering phase for a $[[85,1,7]]$ surface code.
    A number of $\M{Z}$ errors occurred on the lattice resulting in defects depicted as red circles.
    Different clusters are depicted with different colors, and adjacent defects are connected by edges of the corresponding color.
    a) Four $\M{Z}$ errors with $n_\mathrm{d} = 2$ and $R_{\text{sph}} = 4$, resulting in one single cluster.
    b) Four $\M{Z}$ errors with $n_\mathrm{d} = 3$ and  $R_{\text{sph}} = 3$, resulting in two different clusters.
    c) Four $\M{Z}$ errors with $n_\mathrm{d} = 4$ and $R_{\text{sph}} = 3$, resulting in three different clusters.
    d) Four $\M{Z}$ errors with $n_\mathrm{d} = 7$ and $R_{\text{sph}} = 1$, resulting in four different clusters.
    \rev{Note that two defects can belong to the same cluster if they both lie within the region where their corresponding bubbles intersect, or if a third defect exists such that the distance between each defect and the third defect is less than or equal to the radius.}}
 	\label{Fig:sphereClustering}
\end{figure*}
%

\rev{Since surface codes are \ac{CSS} codes, we will, for the sake of simplicity, consider the $\M{Z}$ decoder, i.e., the primal lattice where the sites constitute the $\M{X}$ generators.}
The same reasoning can be applied to the dual lattice to implement the $\M{X}$ decoder. 


\subsection{Bubble Radius Evaluation}
\label{sec:SRE}

The primary goal of this stage is to determine the optimal radius to cluster the defects on the lattice. 
Specifically, each defect is placed at the center of a bubble with uniform size. 
This step is critical because defects located in different bubbles will never be paired in the final solution.

From \eqref{eq:error_probWithBetaApprox}, we observe that the performance of an $[[n, k, d]]$ code, with $t = \lfloor(d-1)/2 \rfloor$, is asymptotically determined by the fraction of errors of weight $w = t + 1$ that the decoder can correct.
The main idea is to exploit the number of defects to determine the smallest possible radius of the bubbles such that, given any error pattern of weight $\leq t + 1$, the clustering is carried out in such a way that it never compromises the error correction capability $t$ and aids as much as possible the decoding of $t+1$ error patterns.
\rev{For instance, in the case of a $t + 1$ chain of adjacent errors, resulting in only a pair of defects, we would like to have them inside the same cluster.
This is motivated by the fact that having a single defect inside a cluster leads to the decoding strategy which connects it to its nearest boundary. 
This greedy strategy impairs the error correction capability and we decided to avoid it in the design.}
For this reason, we have that the radius must be at least $R_{\text{sph}} = t + 1$ when the number of defects is $n_\mathrm{d} = 2$.
This ensures that the two defects belong to the same cluster, allowing them to be connected in subsequent steps of the decoding process.
\rev{Therefore, a simple but suboptimal solution is to set this radius to $R_{\text{sph}} = t + 1$.
However, if there are more than two defects in the lattice, this radius can be reduced to potentially improve the error
correction capability of errors with weight $t+1$.}
As a simple example, adopting this choice, both the error patterns in Fig.~\ref{Fig:errorsNONc} would form a single cluster, leading to the performance of one of the decoders presented in \cite{ForValChi:24STM}.

In general, the radius should be determined by evaluating whether two defects need to be connected or not in the solution. 
For clarity, let us first consider a lattice with three defects. 
When focusing on two of these defects, they must be connected if their distance is at most $t + 1$ minus the weight of the edges connecting the remaining defect to a boundary. 
This ensures that the error correction capability is not compromised.
Since there is a single unpaired defect, in the worst-case scenario, it could be connected to a boundary via an edge of weight one, resulting in a bubble radius of $t$.
This scenario represents the most conservative case, where the radius is the largest possible, leading to fewer clusters overall.
Indeed, if a single cluster is generated, the error correction capability becomes similar to that of the \ac{RFire} decoder.
Similarly, for the case with four defects, two defects must be connected if their distance is at most $t + 1$ minus the weight of the edges connecting the other two defects.
In the worst-case scenario, these defects are connected to each other via an edge of weight one, which results again in a radius of $t$.
Generalizing this trend, we observe that from the maximum value of the radius $t + 1$ it is possible to subtract an edge of weight one for each pair of defects, apart from the initial pair that is being analyzed, to determine if they should be connected.
Therefore, the minimum value of the radius, while satisfying the conditions outlined above, will be $t + 1$ minus $\left \lceil{ n_\mathrm{d} / 2} \right \rceil - 1 $, leading to
\begin{align}
\label{eq:Rmax}
R_{\text{sph}} = t + 2 - \left \lceil{\frac{n_\mathrm{d}}{2}} \right \rceil \,.
\end{align}

\subsection{Bubble Clustering Phase}
\label{sec:SCP}

\input{Figures/Algo/BubbleClusteringPhase}

During this stage, utilizing the information provided by the syndrome $\V{s}$ along with the calculated $R_{\text{sph}}$, all defects are organized into one or more clusters. 
In doing so, a subset of edges from the lattice is assigned to each cluster, ensuring the formation of trees, thus preventing the occurrence of any cycles. 
To formalize the algorithm, let us define $\V{v}$ as \rev{the column vector of size $n_\mathrm{d}$} containing the indexes of each defect, and $\M{L}$ as the matrix containing in the $i$-th row the list of the defects belonging to the $i$-th cluster.
The size of $\M{L}$, $N_c \times n_\mathrm{d}$, is determined during the clustering phase depending on the number of clusters needed.
The procedure unfolds as follows.
During the initialization step, the first defect $v_1$ is added to the first row of $\M{L}$, referred to as $\V{l}_1$.
Subsequently, the algorithm systematically examines whether any other defect lies within a distance less than or equal to $R_{\text{sph}}$ from $v_1$.
If, for instance, this condition is verified for the $j$-th defect, then $v_j$ is included in $\M{l}_1$. 
The proposed goal can be effectively achieved employing the function $\texttt{evalD}$, which computes the distances between two vertices in the grid, calculated based on their integer coordinate positions, enabling precise and efficient measurement of spatial relationships.
In particular, defining as $q_i$ and $r_i$ as the quotient and the remainder of the integer division between $s_i$ and $2t$,  the function \texttt{evalD}($s_i,s_j$) returns $|q_i - q_j| + |r_i - r_j|$.
Additionally, we employ the function $\texttt{adjDef}$ to record that defects $v_1$ and $v_j$ are adjacent to each other using an $n_\mathrm{d} \times n_\mathrm{d}$ adjacency matrix $\M{A}$.
Then, the same procedure is sequentially applied to all the other defects added to $\M{l}_1$. 
Starting from the second defect of each cluster, we have to verify whether a particular defect is already present in the corresponding list before adding it.
This can be performed by maintaining an array storing the cluster membership of each defect, thereby preventing cycles.
Note that a defect that is already part of a cluster, meaning it is adjacent to at least one other defect, can only become adjacent to additional defects when considering the bubble centered around that defect. 
This is essential to ensure that the algorithm produces a tree structure.
For each defect, \rev{the column vector $\V{p}$ of size $n_\mathrm{d}$} stores the index of the cluster to which it belongs.
\rev{Also, the cardinality of each cluster is stored in a column vector $\V{c}$ of size $N_c$.}
Once all defects in $\M{l}_1$ have been processed, if there are any remaining defects not yet included in the list, the first of these defects is added to $\M{l}_2$, and a new cluster is initialized.
This process iterates until all defects have been assigned to a cluster, resulting in a total of 
$N_c$ clusters.
\rev{An instance of the bubble clustering procedure, illustrating the defect graph, is shown in Fig.~\ref{Fig:sphereClustering2}.
For clarity, we will henceforth use the lattice representation notation.}
Some examples of bubble clustering procedures for the $[[85,1,7]]$ surface code are depicted in Fig.~\ref{Fig:sphereClustering}.

Note that, with this procedure, once a cluster is formed, no additional defects can be added to it.
As a result, there is no need for time-consuming processing to merge defects from different clusters.
This bubble clustering stage is outlined in Algorithm~\ref{algo:SCP}.
During the algorithm execution, a vector $\V{o}$ is initialized to record the order of each defect.



\subsection{Peeling Phase}
\label{sec:PP}

\input{Figures/Algo/Peeling}

The goal of this step is to begin with the obtained tree and generate a suitable matching for each cluster. 
Then, in the final post-processing phase, each matching will be evaluated using specific metrics, after which it will either be selected as the final solution or set aside for further consideration. 
In the latter case, another solution is generated and an additional peeling phase will be required for that cluster \rev{(see Section~\ref{subsec:EC})}.
Specifically, the peeling phase, described in Algorithm~\ref{algo:PP}, consists of two distinct sections: a pre-processing stage and a matching stage.
Note that the processing described below is performed independently for each cluster, where a tree has been generated by the previous step if more than one defect resides in the cluster.
In the following, we will define the order of a defect as the number of its adjacent defects in the tree of its cluster.
A defect of order one is called \emph{boundary} defect.

\input{Figures/Algo/GhostAdd}

\subsubsection{Ghost Ancillas Addition}
This step is required to assure that each cluster tree comprises an even number of defects and then is applied for all clusters with an odd number of defects.  
Indeed, since surface codes feature non-periodic boundaries, it is possible for an error chain to terminate at a boundary, resulting in the creation of just a single defect.
In such cases, obtaining a suitable match is not feasible.
Therefore, if the cardinality of a cluster is odd, a ghost ancilla is added to one of its defects according to the following criteria.
The required ghost ancilla is connected to the defect closest to a boundary.
If multiple defects are equidistant, the ghost ancilla is attached to the defect that is farthest from the other defects in the cluster.
In particular, the defect that is farthest from the others in the cluster is the one with the maximum distance from its cluster nearest neighbor.
Additionally, we retain in memory the specific boundary (i.e., left or right) to which the ghost ancilla is attached.
In case of a second peeling solution, if the $i$-th cluster contains an even number of defects, we include a ghost ancilla on both the right and left boundaries. 
This is done to preserve the even parity of the defects.
On the other hand, if the number of defects is odd, we add a ghost ancilla to the right (left) boundary if the cluster was attached to the left (right) boundary during the previous peeling phase.
The criteria for determining which defects these ghost ancillas connect to remain the same.
For conciseness, we utilize the function \texttt{addGhost} which takes as input the cluster index and returns the index of the adjacent defect as previously outlined (see Algorithm~\ref{algo:ghost}).
\rev{We remark that the algorithm handles defects sequentially from left to right and top to bottom. As a result, when two defects are equidistant from the boundary and have the same distance from the nearest defect, the ghost ancilla is always assigned to the first one encountered in this order.}
Note that all the required information can be stored during the clustering phase. 
Here it has been presented separately for the sake of presentation clarity.

\input{Figures/Algo/BuildMatch}

\subsubsection{Tree Peeling}
In this stage, for each cluster $i$-th, the corresponding tree is iteratively peeled to achieve a first suitable matching $\mathcal{E}_i^{1}$.
For this purpose, we define as vertices the defects present in the current cluster, all initially set to be switched on (i.e., \rev{set} $s_i$ to one).
For each cluster, there are two possibilities: it could contain either an even or an odd number of vertices. 
If the number is odd, we must add edges between the vertex selected by \texttt{addGhost} and the nearest boundary to complete the final matching.
\rev{After this,} we switch off (i.e., \rev{set} $s_i$ to zero) the vertex selected by \texttt{addGhost}.
In this way, we have ended up to an even number of vertices, allowing the algorithm to proceed identically for both possibilities.
In Algorithm~\ref{algo:PP}, we employ the function $\texttt{buildMatch}$ to describe this matching.
This function takes as input the indices of two vertices and inserts all the qubit edges between them into the final solution (i.e., the qubit in the path composing the edge of the tree). 
On the other hand, if the function is called with only one single index, it adds the edges between the input vertex and the nearest boundary to the final matching.
In this way, we connect it to the ghost ancilla in that cluster.
This can be efficiently handled by using simple modular arithmetic operations, which exploit the regular structure of the lattice to directly compute qubit indices without requiring complex lookups or additional overhead.

After this pre-processing, the algorithm processes each boundary vertex (i.e., a vertex of order one) within the cluster as follows. 
To identify it, we employ the vector $\V{o}$, prepared during the clustering phase, \rev{which contains the order of each defect.}
For each of the boundary vertices, we have one adjacent vertex indexed by \texttt{idxAdj}, retrieved using the function \texttt{Adjacent}.
Then, if the current boundary vertex is switched on, we add its incident edge to the final matching $\mathcal{E}_{i}^{1}$, and its adjacent vertex is flipped (i.e., $s_i \leftarrow 1 - s_i $). 
Afterward, even if the current boundary vertex was switched off, we remove the incident edge from the graph and update the vertex orders accordingly.
This is accomplished using the function \texttt{peel}. 
Finally, before proceeding with the next boundary vertex, we also update the cardinality of the cluster $\V{c}[i]$.
The desired matching among the initial defects is achieved as soon as the tree becomes an empty graph, which terminates the procedure.
In addition, we keep track of the number of physical qubits associated with each matching. 
This count is referred to as the weight $w_i^{1}$ of the matching $\mathcal{E}_{i}^{1}$.
The peeling procedure is detailed in Algorithm~\ref{algo:PP}.
The functions \texttt{addGhost} and \texttt{buildMatch} are detailed in Algorithm~\ref{algo:ghost} and Algorithm~\ref{algo:match}\rev{, respectively.} 

We remark that this peeling algorithm is highly efficient, especially when compared to the \ac{UF} decoder. 
In our approach, only the edges connecting adjacent defects in the tree are stored. 
In contrast, the \ac{UF} decoder expands clusters in all directions across the lattice to match parity, requiring the storage and processing of a much larger number of edges.
An example of peeling phase procedure, for the error pattern of Fig~\ref{Fig:sphereClustering}b, is described in the following example.
\begin{example}
In Fig~\ref{Fig:TreePeel} we report an example of peeling phase procedure.
The arrows indicate the currently processed defect, while switched-off defects are represented as white circles.
Vertices are processed sequentially from left to right and top to bottom during the operation.
If a ghost ancilla is present, it will always be processed first.
Edges included in the solution are highlighted in green. 
a) Radius evaluation phase.
b) Bubble clustering phase and addition of a ghost ancilla where necessary.
c) A boundary vertex for the orange cluster is identified.
Since this vertex is currently switched on, it will be switched off along with its adjacent defect.
The incident edge will also be added to the solution.
The edge connected to the ghost ancilla is added into the blue solution, and the adjacent vertex is switched off.
d) A boundary vertex for the blue cluster is identified. 
Since this vertex is currently switched on, it will be switched off along with its adjacent defect. 
The incident edge will also be added to the solution.
e, f) A boundary vertex is identified.
Since this vertex is switched off, the incident edge will be discarded.
g) A boundary vertex for the blue cluster is identified.
Since this vertex is currently switched on, it will be switched off along with its adjacent defect.
The incident edge will also be added to the solution.
h) Resulting matching for both clusters.
\end{example}

\begin{figure*}[t]
 	\centering
     \includegraphics[width = 0.99\textwidth]{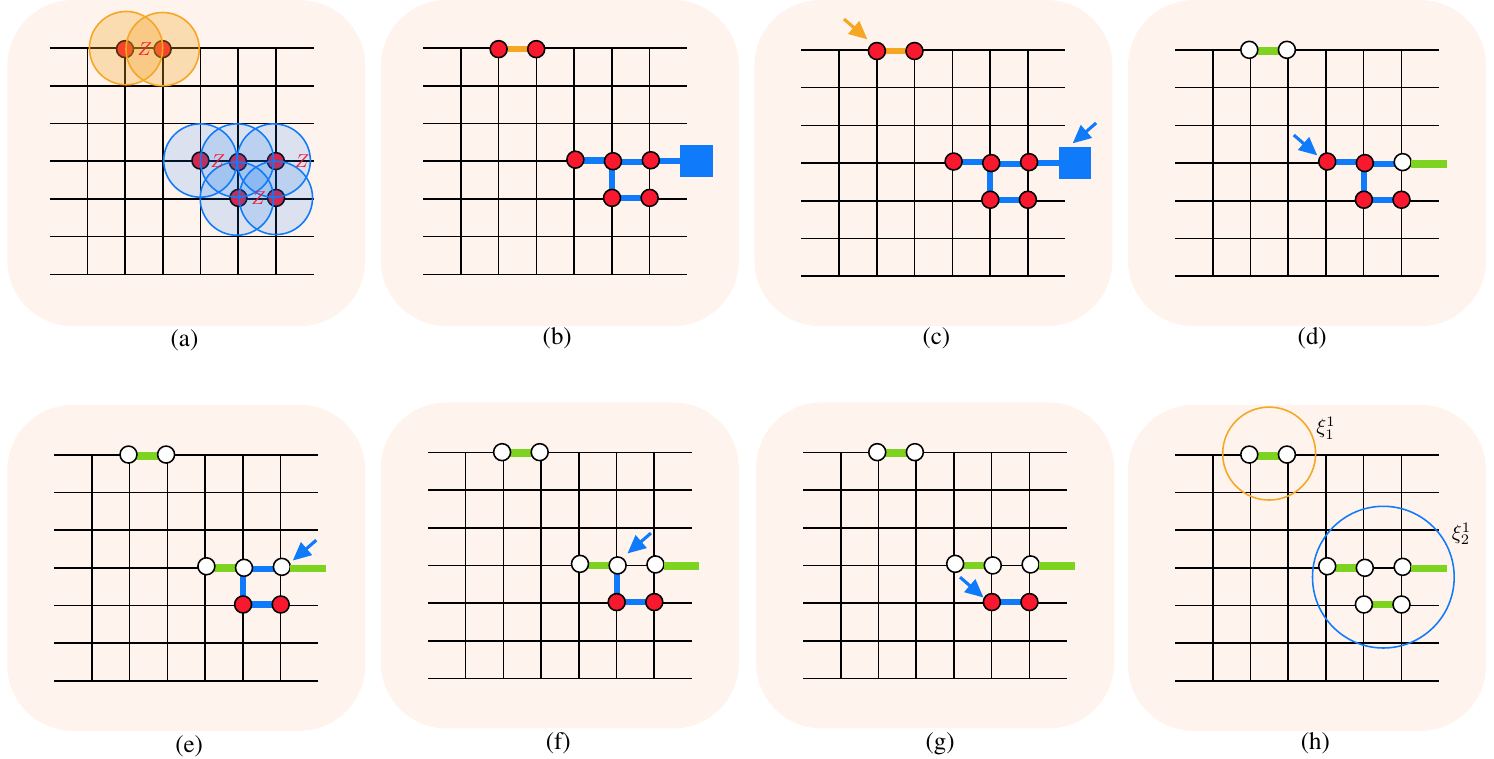}
 
 \caption{A detailed step-by-step example of the peeling phase procedure on a $[[85,1,7]]$ surface code with four errors present in the lattice. }
 	\label{Fig:TreePeel}
\end{figure*}

\subsection{Post-Processing Phase}
\label{subsec:EC}

The goal of this phase is to identify a suitable error correction pattern for each cluster, utilizing the matching solutions obtained from previous stages.
For each cluster, if $w_i^{1} \leq t$, the corresponding matching is selected as the error pattern and this phase is skipped.
On the other hand, if the $i$-th cluster does not satisfy this condition, a post-processing is required (i.e., an additional peeling phase).
The main idea is to construct a new solution differing from the previous one by a logical operator. 
In this way, we can compare the two solutions using our tailored metric.

Then, the tree peeling is performed on this new graph to obtain the matching $\mathcal{E}_i^{2}$.
The error pattern is chosen as follows: 
\begin{itemize}
    \item if $w_i^{2} \leq t$, $\mathcal{E}_i^{2}$ is chosen as the current solution;
    \item else if $w_i^{1} = t + 1$, $\mathcal{E}_i^{1}$ is chosen as the current solution;
    \item else if $w_i^{2} = t + 1$, $\mathcal{E}_i^{2}$ is chosen as the current solution.
\end{itemize}
At this point it is very likely that the solution has been already determined. However, if both $w_i^{1}$ and $w_i^{2}$ have weight $> t+1$, a different decision criterion, borrowed from \cite{ForValChi:24STM}, is applied.
Let us define a \emph{column} as the set of horizontal edges aligned vertically on the lattice, and let us enumerate columns from left to right, ranging from $1$ to $d$.
Then, considering cluster $i$-th, we define two vectors $\V{c}_i^{j}$ with entries $c_i^{j,k}$, where $j=1, 2$ and $k=1, \dots, d$, representing the cardinality of the intersection between $\mathcal{E}_i^{j}$ and the $k$-th column. Also, we define as $\V{u}_i^{j}$ a vector with entries $u_i^{j,k} = c_i^{j,k} \mod 2$.
Hence, the final error pattern consists in the matching that minimize
\begin{align}
    \label{eq:metricCol}
    f(\mathcal{E}_i^j) = \sum_{k=1}^{d} u_i^{j,k}\,.
\end{align}
This decision criterion consistently preserves the distance of the code~\cite{ForValChi:24STM}.

\subsection{Adjustments for High Physical Error Tolerance}\label{sec:PhyErr}
Previously, we focused our attention on guaranteeing the error correction capability $t$ and obtaining a good performance in the low physical error region given by error patterns of weight $t+1$.
Here, we present some simple adjustments for the bubble clustering phase which mainly target error patterns of weight greater than $t+1$, affecting the performance curve in the high physical error regime.

\subsubsection{Setting the Minimum} Bubble Radius
The \ac{BC} decoder, as described in previous sections, achieves asymptotic performance that is very close to that of the standard \ac{MWPM} decoder. 
However, some simple precautions are necessary when dealing with high physical error rates, specifically when the number of defects exceeds $2t + 2$.
In particular, according to \eqref{eq:Rmax}, the maximum radius of a cluster would be zero, making it impossible to connect any pair of defects and always relying on boundary connections.
Hence, to avoid failures in the decoding process caused by this issue, we adjust the bubble radius as 
\begin{align}
\label{eq:RmaxHE}
R_{\text{sph}} =  \begin{dcases}
    t + 2 - \left \lceil{\frac{n_\mathrm{d}}{2}} \right \rceil & \text{if}~n_\mathrm{d} \leq 2t, \\
    2 & \text{otherwise}.\\
\end{dcases}
\end{align}
Note that, using \eqref{eq:RmaxHE}, not only the radius cannot be zero, but also it cannot be one.
This has been done to guarantee that error patterns with a double error occurring on two adjacent qubits in the actual lattice reside always in the same cluster if all other errors are not interfering with their defects.
Beside this reasoning, we also fine-tune this adjustment testing other possibilities, such as constant radius one and non-constant shapes.
The solution in \eqref{eq:RmaxHE} was the one having the best performance for error patterns of weight greater than $t+1$.
In Fig.~\ref{Fig:plot_adj} we show the impact of this adjustment on a $[[85,1,7]]$ surface code. 
We observe that, for high physical error rates, applying these radius modifications provides an advantage in terms of the logical error rate, as demonstrated by comparing the black and yellow curves. 
This improvement arises from the enhanced ability to correct errors of weight $ \geq t + 1$.


\subsubsection{Star-Defects Avoidance}
Recall that the bubble clustering stage, described in Section~\ref{sec:SCP}, is performed on the defect graph. 
This implies that there is no direct correspondence between the edges in the resulting clusters and the physical qubits in the lattice.
Hence, since the defects are processed sequentially, due to $R_\mathrm{sph}$ it is possible that several defects are connected to the defect under processing even if a different defect was the nearest neighbour, creating a \emph{star}-like graph.
These star defects do not harm the decoder ability to correct up to the \rev{code distance}, however, they worsen the performance in the high physical error regime (capability to correct error pattern with more than $t$ errors).  
For instance, let us consider the bubble clustering stage in Fig.~\ref{Fig:postProc}a.
Since the error causes five defects, $R_{\text{sph}} = 2$.
Then, according to the standard procedure and proceeding from the lower vertex index to the higher one, the defect $v_2$ results in a star defect, being adjacent to defects $v_1$, $v_3$ and $v_4$. 
In this way, the resulting tree is the one having as edges: $(v_1, v_2)$, $(v_2, v_3)$, $(v_2, v_4)$, and $(v_3, v_5)$.
Indeed, during the peeling phase, this tree results in a matching consisting of: a weight-two horizontal edge between the ghost ancilla and $v_1$; a weight-two edge between $v_2$ and $v_4$; and a weight-two edge from $v_3$ to $v_5$.
In this particular configuration, the problem arises if the \texttt{buildMatch}($v_2$,$v_4$) does not share a qubit with \texttt{buildMatch}($v_3$,$v_5$) (e.g., the qubit between $v_3$ and $v_4$).
In fact, if this qubit is not shared, the weight results in  $w_1^1(\mathcal{E}_1^1) = 6$.
Moreover, the alternative solution, obtained by attaching a ghost ancilla to $v_5$ instead of $v_1$ would have weight $w = 6$.
In this case, both solutions are greater than $t+1$. Furthermore, according to metric \eqref{eq:metricCol}, the incorrect solution would be selected, causing an error in the correction process.

To avoid this, we introduce an improvement of the bubble clustering.
Specifically, each time a defect is added to a cluster, we record the distance to its adjacent defect within the tree.
Then, if the defect being processed is closer to another defect already in the tree, and both defects are adjacent to the same star defect, the second defect is detached from the star defect and attached to the defect being processed.
This modification requires additional processing, which may increase the execution time. 
However, it should be noted that, in the presence of star defects, the information about the same physical qubits must be retrieved multiple times from memory during the peeling phase.
Therefore, avoiding star defects helps recover execution time in this regard.
Referring back to the previous example, with the proposed modification, when processing defect $v_3$, it becomes possible to detach $v_4$ from $v_2$, and attach $v_3$ to $v_4$.
Moreover, when processing $v_5$ it is also possible to detach it from $v_3$ and attach it to $v_4$, obtaining the tree shown in Fig.~\ref{Fig:postProc}b.
In this way, the first matching has weight $w_1^1(\mathcal{E}_1^1) = 4$, and will be chosen as a final solution, correcting the error.
We remark that this adjustment is also able to improve the asymptotic performance (i.e., correction of weight $t+1$ error patterns). 
Let us clarify this procedure with an example. 
\begin{example}
  In Fig.~\ref{Fig:postProc}, we have four $\M{Z}$ channel errors occurred on the lattice, resulting in defects depicted as red circles. 
  Since $n_\mathrm{d} = 5$, the bubble radius $R_{\text{sph}} = 2$, resulting in a single cluster with an odd number of defects, necessitating the insertion of a single ghost ancilla. 
  a) Bubble clustering decoding without star-defects avoidance. The matching $\mathcal{E}_1^{1}$, evaluated during the peeling phase, has weight $w_1^{1} = 6 > t$.
  Therefore, a post-processing phase is required.
  The second matching has weight $w_1^{2} = 6 > t$.
  Both solutions have weight $w_1 > t + 1$.
  Therefore, the solution is selected based on \eqref{eq:metricCol}, which results in choosing the faulty error pattern.
  b) Bubble clustering decoding with star-defects avoidance. The matching $\mathcal{E}_1^{1}$, evaluated during the peeling phase, has weight $w_1^{1} = 4 > t$.
  Therefore, a post-processing phase is required.
  The second matching has weight $w_1^{2} = 6 > t$.
  Moreover, the weight of the first matching is equal to $t + 1$.
  Hence, $\mathcal{E}_1^{2}$ is chosen as final matching, and the error is actually corrected.
\end{example}

\begin{figure}[t]
 	\centering
\includegraphics[width=0.99\columnwidth]{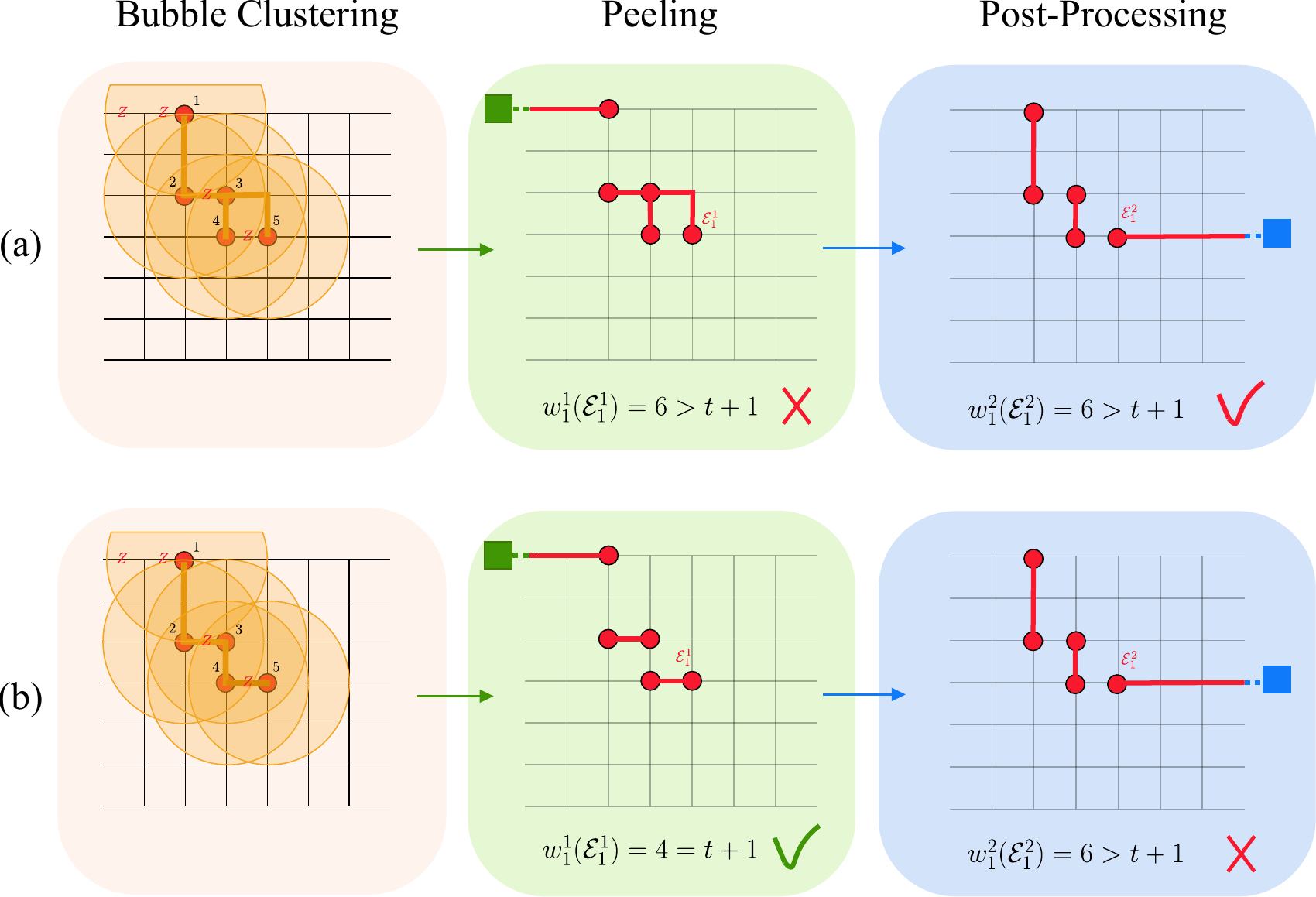}
 	\caption{
  Examples of \ac{BC} decoding for a $[[85,1,7]]$ surface code:
a) without star-defects avoidance.
b) with star-defects avoidance.}
 	\label{Fig:postProc}
\end{figure}
In Fig.~\ref{Fig:plot_adj} we show the impact of this adjustment on a $[[41,1,5]]$ surface code. 
Without star-defect avoidance, the logical error rate is represented by the yellow curve. 
However, by applying these techniques, we achieve the green curve, indicating a significant improvement. 
This gain shows that avoiding star defects enhances the correction of errors with weight $\geq t + 1$, leading to improved overall performance.
For completeness, we also include the case where the star-defect avoidance technique is employed without radius adjustments, shown by the red curve.
The radius adjustment technique yields performance improvements at high physical error rates, whereas the star-defect avoidance technique provides an advantage at low physical error rates.

Adopting the described techniques, we observe an overall improvement of the performance that increases with the code distance.
As an example, in Fig.~\ref{Fig:plot_adj}, for $p = 10^{-2}$, we observe an improvement of $90\%$ in error correction.
For $p = 10^{-2}$, the $[[41,1,5]]$ surface code yields an improvement of $42\%$, while the $[[145,1,9]]$ surface code achieves a higher improvement of $78\%$.

\subsubsection{High \rev{Code Distance}}

As the number of qubits scales with the square of the \rev{code distance}, in this error regime it is beneficial to implement additional low-complexity measures to refine the error correction capability.
Specifically, for surface codes with a \rev{code distance} of $d\geq11$, we will adopt the following measures:
\begin{itemize} 
\item During the sphere clustering phase, we record clusters that contain a single defect. 
If exactly two such clusters are found, and their defects are separated by a distance of $R_{\text{sph}}^{\prime} + 1$, we merge them into a single cluster containing both defects.
\item If a cluster contains a single defect, and its distance from the boundary equals the distance to a defect in an odd-cardinality cluster, we connect these defects in a single cluster.
\end{itemize}


\section{Bubble Clustering Decoder: Analysis} \label{sec:complexity}

\subsection{Error Correction Capability Preservation}
Before proceeding with the last technical details of the decoder, we demonstrate that the whole processing using the \ac{BC} decoder ensures error correction capability up to the \rev{code distance}.

\begin{figure}[t]
	\centering
    \includegraphics[width=0.99\columnwidth]{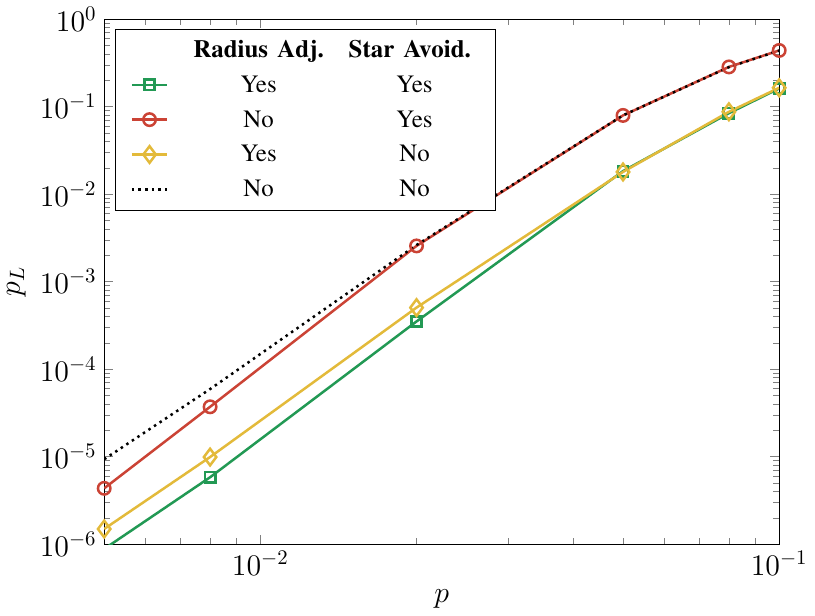}
	\caption{Logical error rate versus physical error rate of the $[[85,1,7]]$ surface code. The curves illustrate the impact of radius adjustments and star-defect avoidance techniques on the error correction capability.
		\label{Fig:plot_adj}}
\end{figure}

\begin{theorem}
\label{th:DistPreserved}
For an $[[n, k, d]]$ surface code, an error pattern with weight $w \leq t$ is always correctly corrected using the \ac{BC} decoder.
\end{theorem}

\begin{proof}
Note that in case both $\mathcal{E}_i^{1}$ and $\mathcal{E}_i^{2}$ have weight~$>~t + 1$, \eqref{eq:metricCol} is employed, and the error is always corrected as proven in~\cite{ForValChi:24STM}. 
Regarding errors of weight $\leq t$, for which at least one matching between $\mathcal{E}_i^{1}$ and $\mathcal{E}_i^{2}$ has weight $\leq t + 1$, we want to show that these are always corrected.
To this aim, we remark that one of the two possible solutions, $\mathcal{E}_i^{1}$ and $\mathcal{E}_i^{2}$, always corrects the error, while the other always realizes a logical operator.
Indeed, the two solutions are related by the application of a logical operator since the surface codes encode $k=1$ information qubit.
Let us initially focus on a surface code with odd distance.
As a worst case, we first consider $t$ Pauli errors in the same row of the lattice, possibly causing an horizontal logical operator of weight $2t + 1$.
In case of an odd number of defects, one of the two solutions $\mathcal{E}_i^{1}$ or $\mathcal{E}_i^{2}$ has weight $w = t$ due to the horizontal disposition, also corresponding to the actual error.
The alternative solution has always weight $2t+1-t = t+1$.
Moreover, if an even number of defects is obtained, the first solution (with no ghost ancillas) always coincides with the error and has weight $t$.
Hence, in both cases, the error is always corrected.
Next, let us consider the case in which $t$ errors occur on two adjacent rows, possibly causing a logical operator of weight $2t + 2$.
In this configuration, it is possible for the correct solution not to be of minimum weight, i.e., its weight can be greater than $t$.
This happens when the final matching results in the correct solution plus an element of the stabilizer. 
However, given that the possible logical operator has a weight of $2t +2$ and the actual error has a weight of $t$, the weight of the incorrect solution will be at least $2t + 2 - t = t+2 > t + 1$. 
Hence, in the worst case, both solutions have weight greater than $t + 1$ and, using \eqref{eq:metricCol}, the error is corrected.
Moreover, if $t$ errors occur across multiple rows, the weight of the resulting logical operator would exceed $2t+2$, further reinforcing the validity of the previous argument.
The same arguments apply to a surface code with even distance, where the corresponding logical operators have a weight greater by one, i.e., the smallest logical operator has weight $2t + 2$.
\end{proof}

This theorem has important consequences on the decoding performance. 
Indeed, there is a great amount of errors of weight $t + 1$ that are not corrected by \eqref{eq:metricCol} but can result in a solution of weight $t + 1$.
Hence, they are actually corrected by the \ac{BC} decoder, according to Section~\ref{subsec:EC}. 
An instance of the \ac{BC} decoder requiring a post-processing phase is depicted in Fig.~\ref{Fig:postProc} for the $[[85,1,7]]$ surface code.

\begin{corollary}
   \label{cor:ErrCLusters}
    Any error pattern of weight $t + \ell - 1$, for which the bubble clustering procedure results in at least $\ell > 0$ distinct clusters, is always corrected.
\end{corollary}

\subsection{Complexity Analysis }

In this section we analyze the complexity of the \ac{BC} decoder. 
The first phase, i.e., the evaluation of the bubble radius, is performed with a complexity of $O(1)$.
\rev{Next, for each defect the algorithm potentially compares it to every other defect to check the distance condition.
This check is performed
during the bubble clustering phase.}
Hence, the worst-case scenario involves $O(n_\mathrm{d}^2)$ operations.
Also, before adding a defect to a cluster, the algorithm checks if it is already in the cluster.
Each check is constant time $O(1)$ because it simply involves looking up the defect's membership in an array.
Thus, the total complexity is dominated by the quadratic terms, resulting in an overall complexity of $O(n_\mathrm{d}^2)$. 
The star-defects avoidance technique involves checking whether two defects share any adjacent defects, and this check is performed during the bubble clustering phase whenever two defects are assigned to the same cluster.
This verification can be performed efficiently by accessing a pre-stored element in an array that was created during the cluster construction specifically for this purpose. 
As a result, this procedure can be completed in constant time, with a complexity of $O(1)$.
The tree peeling stage processes each cluster to iteratively remove defects and edges until a suitable matching is achieved.
Since each defect in the lattice is part of exactly one cluster, and each operation (inverting parity, removing edges) is proportional to the number of defects, the overall complexity is proportional to the total number of defects $O(n_\mathrm{d})$.
Moreover, in the post-processing phase, the weight of each cluster matching is compared to the $t$ parameter.
This comparison is a simple $O(1)$ operation per cluster. 
If there are $N_\mathrm{c}$ clusters and recalling that $N_\mathrm{c} \le n_\mathrm{d}$ the complexity becomes $O(n_\mathrm{d})$.
Finally, the vector $\V{u}_i^{j}$ can be computed simply by enumerating the horizontal qubits in the lattice row by row, from left to right.
It is sufficient to perform a single mod operation (which can be performed in $O(1)$) for each horizontal qubit in the solution.
Thus, the worst-case complexity of this step is $O(w)$, where $w$ is the weight of the error.
Note, however, that this step is required only in a small \rev{number} of cases.
Overall, the \ac{BC} decoder complexity is primarily determined by the bubble clustering stage.
Note that for an $[[n,k,d]]$ surface code, in the case of an error of weight $\leq t + 1$, the number of defects is $n_\mathrm{d} \leq 2t + 2$. Hence, the complexity of the algorithm in this regime can be written as $O(d^2)$ or $O(n)$, considering that in a surface code $n = d^2 + (d-1)^2$.
Table~\ref{tab:TCompl} presents a comparison of the complexities of different decoders.
Here, $n$ represents the number of qubits, $n_\mathrm{d}$ denotes the number of defects, and $\alpha$ is the inverse of Ackermann's function\cite{Del:21}.

\input{Figures/table_cpmplexity}

\section{Numerical Results}\label{sec:NumRes}

In this section we compare the performance of surface codes using \ac{BC}, \ac{MWPM}, \ac{UF}, and \ac{RFire} decoding via Monte Carlo simulations. 
We included \ac{RFire} in the comparison because it offers faster execution times than \ac{STM} while providing similar error correction capability.
All these decoders are implemented in C++, run with an Apple Silicon M2 processor and executed on a single core. 
The \ac{UF} algorithm is a C++ implementation of the efficient weighted union find described in \cite{Del:21}.
In the standard implementation, we exploit the \acs{LEMON} C++ library for an efficient \ac{MWPM} algorithm \cite{DezBalJut:11}.
Also, we employ the PyMatching 2 library for the sparse blossom implementation~\cite{HigGid:23}.
Additionally, when using PyMatching, we submit shots in batches of 1000 to minimize the overhead of Python-to-C++ method calls.
Also, for an efficient implementation of \acs{LEMON} version of the \ac{MWPM}, and \ac{RFire} decoders, we use Manhattan distance to assign weights and construct the graph of defects, avoiding the usage of Dijkstra's algorithm. 
Regarding the \ac{BC} decoder, the numerical evaluation is performed by including all adjustments discussed in Sec.~\ref{sec:PhyErr}. 
\rev{In evaluating the logical error rates, for each simulation point, we have run simulations until observing at least one hundred errors.}

\begin{figure}[t]
	\centering
    \includegraphics[width=0.99\columnwidth]{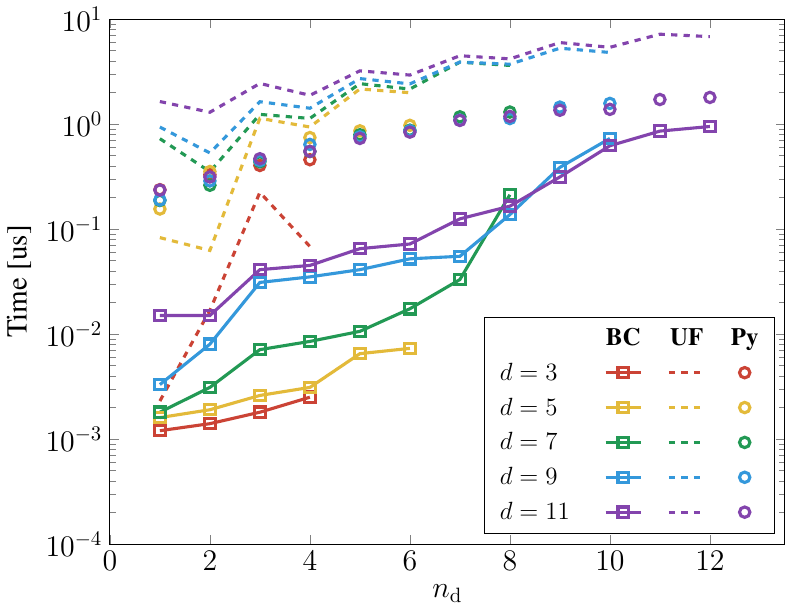}
	\caption{ Average execution times per single decoding versus number of defects. The abbreviation Py stands for PyMatching \rev{\cite{HigGid:23}} version of the \ac{MWPM}. 
		\label{Fig:plot_time}}
\end{figure}

\subsubsection{Average Execution Times}
To assess the complexity of the decoders, we measure the average execution time of the complete decoding procedure: from the syndrome measurement to the final solution.
Fig.~\ref{Fig:plot_time} and Fig.~\ref{Fig:plot_time2} show the average execution times per single decoding over the number of defects in the lattice.
The evaluation is carried out when varying the number of defects and the lattice size, from the $[[13,1,3]]$ to the $[[221,1,11]]$ surface codes.
Specifically, the decoders are provided with batches of 1000 instances, each containing $n_\mathrm{d}$ defects.
From Fig.~\ref{Fig:plot_time}, it is evident that the \ac{BC} decoder outperforms both PyMatching and \ac{UF} algorithms in terms of execution speed across all code instances.
From Fig.~\ref{Fig:plot_time2}, we observe that the \ac{BC} decoder achieves a time savings of over an order of magnitude compared to the \ac{RFire} decoder for code distances greater than three.
Moreover, note that the \ac{RFire} decoder achieves a speedup of over three orders of magnitude compared to the \acs{LEMON} implementation of the \ac{MWPM}, as detailed in~\cite{ForValChi:24STM}.
However, in the present analysis, execution times are slightly longer compared to~\cite{ForValChi:24STM} because we include the graph generation procedure as part of the overall decoding process for all decoders.
Moreover, Fig.~\ref{Fig:plot_time3} presents a comparison of the average execution times of the \ac{BC}, \ac{UF}, and Pymatching algorithms for \rev{high code distances,} i.e., surface codes up to the $[[685, 1, 19]]$. 
The results indicate that the performance advantage of the \ac{BC} decoder remains consistent as the code distance increases.

\begin{figure}[t]
	\centering
    \includegraphics[width=0.99\columnwidth]{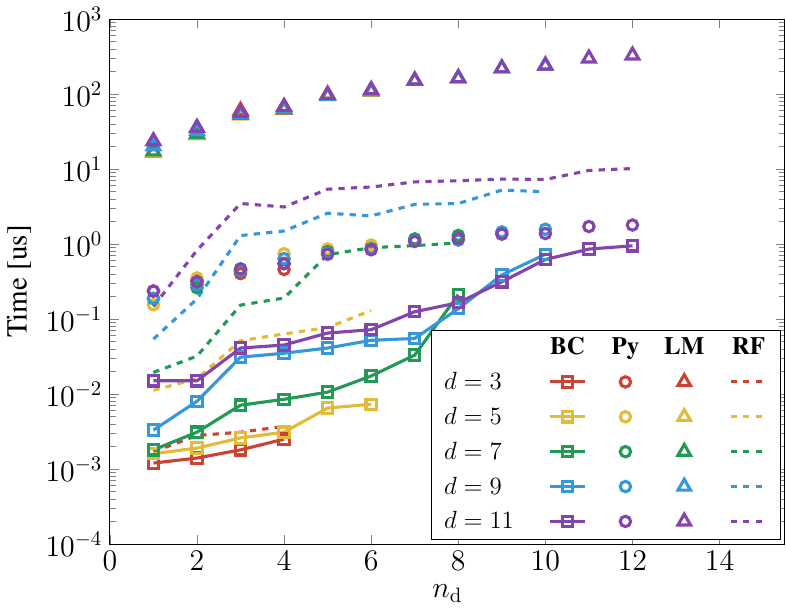}
	\caption{ Average execution times per single decoding versus number of defects. The abbreviations Py and LM stand for PyMatching~\rev{\cite{HigGid:23}} and LEMON versions of the \ac{MWPM}.
		\label{Fig:plot_time2}}
\end{figure}
\begin{figure}[t]
	\centering
    \includegraphics[width=0.99\columnwidth]{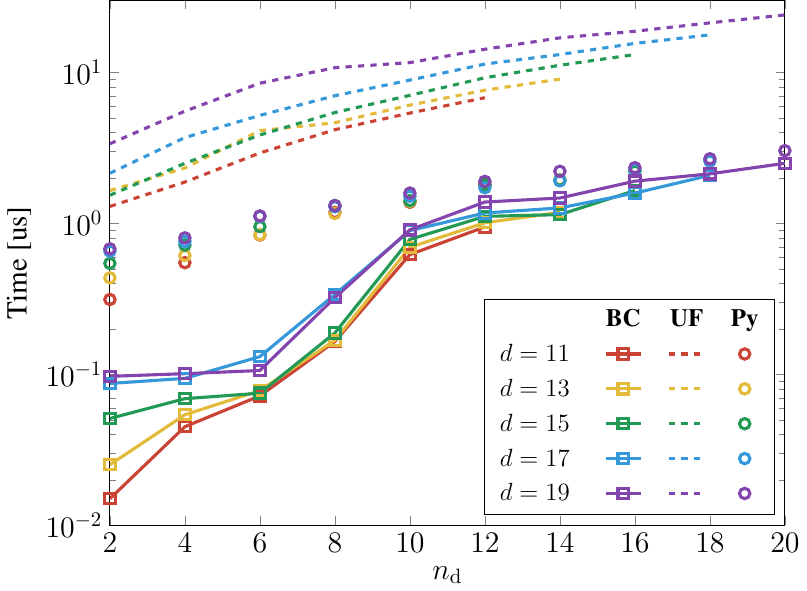}
	\caption{ \rev{Trend of the average execution times \rev{for high code distances}. The abbreviation Py stands for PyMatching~\rev{\cite{HigGid:23}} version of the \ac{MWPM}.}}
		\label{Fig:plot_time3}
\end{figure}

\subsubsection{Logical Error Rate}
In Fig.~\ref{Fig:plot_perf} and Fig.~\ref{Fig:plot_perf2}, we show the logical error rate as a function of the physical error rate of some surface codes over depolarizing channel.
We can observe that, employing the clusterization procedure, we have a large gain in error correction capability for the \ac{BC} decoder when compared to the \ac{RFire}. 
In particular, the proposed decoder is able to correct a much greater fraction of error patterns of weight $\geq t + 2$.
\rev{From these results we observe that, for shorter code distances, the performance of the \ac{BC} and the \ac{MWPM} decoders are quite comparable.}
In Fig.~\ref{Fig:plot_perf2} we observe that the performance gap between the \ac{MWPM} and \ac{BC} decoders widens when comparing distances $d = 7$ and $d = 9$, but shows only a slight increase between $d = 9$ and $d = 11$.

\begin{figure}[t]
	\centering
    \includegraphics[width=0.99\columnwidth]{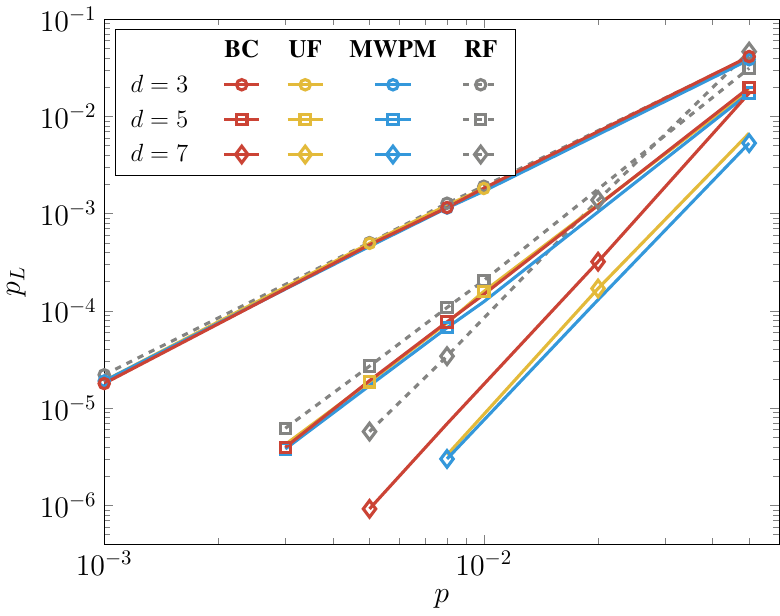}
	\caption{ Logical error probability versus physical error probability of the channel  over depolarizing channel.
		\label{Fig:plot_perf}}
\end{figure}
\begin{figure}[t]
	\centering
    \includegraphics[width=0.99\columnwidth]{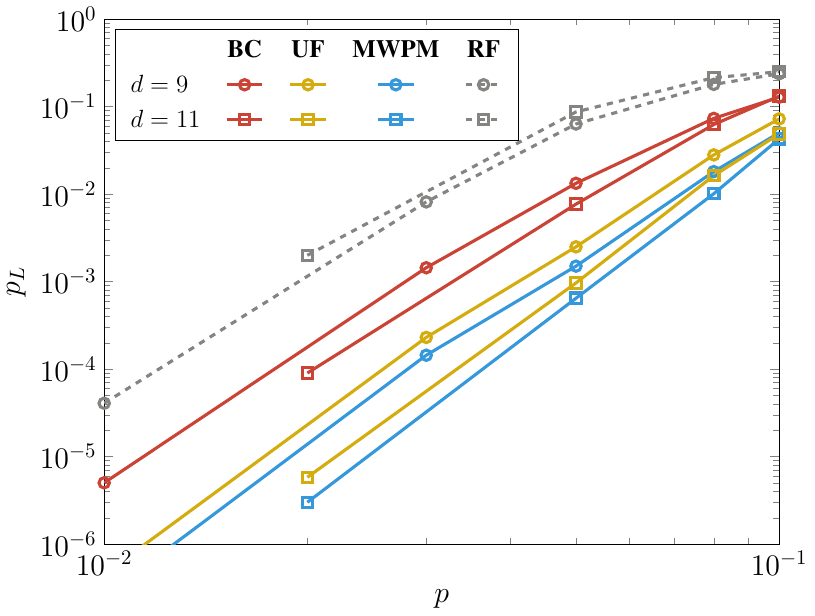}
	\caption{ Logical error probability versus physical error probability of the channel  over depolarizing channel.}
		\label{Fig:plot_perf2}
\end{figure}
%


\section{Conclusions}\label{sec:conclusions}

In this work, we introduce a novel decoder for quantum surface codes that prioritizes speed, \rev{accepting a trade-off in error-correction performance}. After demonstrating its ability to correct all error patterns with weight $\leq t$, we then shift our focus to more complex, higher-weight error patterns. Our decoder consistently achieves execution times in the sub-microsecond range for practical surface code lattice sizes. Notably, its complexity scales linearly with the number of qubits, provided that the number of errors remains within the decoder's design limits. We benchmark our proposal against several established quantum decoders, including the \acs{LEMON} and PyMatching implementations of the \ac{MWPM}, the \ac{UF}, and the \ac{RFire}.



\bibliographystyle{IEEEtran}
\bibliography{Files/IEEEabrv,Files/StringDefinitions,Files/StringDefinitions2,Files/refs}

\end{document}

%% file: Files/Acronimi_SICMMA.tex
\begin{acronym}
\small
\acro{AWGN}{additive white Gaussian noise}
\acro{BCH}{Bose–Chaudhuri–Hocquenghem}
\acro{BC}{bubble clustering}
\acro{CDF}{cumulative distribution function}
\acro{CRC}{cyclic redundancy code}
\acro{LDPC}{low-density parity-check}
\acro{LUT}{lookup table}
\acro{ML}{maximum likelihood}
\acro{MWPM}{minimum weight perfect matching}
\acro{QECC}{quantum error correcting code}
\acro{PDF}{probability density function}
\acro{PMF}{probability mass function}
\acro{MPS}{matrix product state}
\acro{WEP}{weight enumerator polynomial}
\acro{WE}{weight enumerator}
\acro{BD}{bounded distance}
\acro{QLDPC}{quantum low density parity check}
\acro{CSS}{Calderbank, Shor, and Steane}
\acro{MST}{minimum spanning tree}
\acro{PruST}{pruned spanning tree}
\acro{RFire}{Rapid-Fire}
\acro{UF}{union-find}
\acro{LEMON}{library for efficient modeling and optimization in networks}
\acro{STM}{spanning tree matching}
\acro{i.i.d.}{independent identically distributed}
\acro{QEC}{quantum error correction}
\acro{BP}{belief propagation}

\end{acronym}

%% file: Figures/Algo/BubbleClusteringPhase.tex
\begin{algorithm}[t]
\SetKwInOut{Input}{input}
\SetKwInOut{Output}{output}
\caption{Bubble Clustering Phase} \label{algo:SCP}
\Input{$n_\mathrm{d}$, $\V{v}$, $R_{\text{sph}}$; \\} 
\Output{$\M{A}$, adjacency matrix; \\
$\V{p}$, indexes of the clusters for each defect;\\
$\V{c}$, cardinality of each cluster; \\ $\M{L}$, matrix of clusters; \\ $\V{o}$, order of each defect; \\  } 
\BlankLine
init $\M{A}$, $\V{p}$, $\V{c}$ to all zeros; \\
$N_c \gets 0$; \\
$\mathrm{totalDef} \gets 0$; \\
\While{$\mathrm{totalDef} \leq n_\mathrm{d}$}{
$\mathrm{contDefect} \gets 1$; \\
\While{$\V{p}[\mathrm{contDefect}] \neq 0$}{$\mathrm{contDefect} \gets \mathrm{contDefect} + 1$;\\}
$N_c \gets N_c + 1$; \\
$\mathrm{totalDef} \gets \mathrm{totalDef} + 1$; \\
$\V{c}[N_c] \gets \V{c}[N_c] + 1$; \\
$\M{L}[N_c][\V{c}[N_c]] \gets \V{v}[\mathrm{contDefect}]$; \\
$\V{p}[\V{v}[\mathrm{contDefect}]] \gets N_c$; \\
$\mathrm{flagEndCluster \gets 0}$; \\
$\mathrm{currentDef} \gets 1$; \\
\While{$\mathrm{flagEndCluster = 0}$}{
$\mathrm{flagEndCluster \gets 1}$; \\
\ForAll{$n_{\mathrm{d}_i} \in \V{v}$}{
\If{$\texttt{evalD}(\M{L}[N_c][\mathrm{currentDef}], n_{\mathrm{d}_i}) \leq R_{\text{sph}}$}
{
\If {$\V{p}[n_{\mathrm{d}_i}] = 0$}{
$\M{A} \gets \texttt{adjDef}(\mathrm{currentDef},n_{\mathrm{d}_i}) $; \\
$\V{o}[\M{L}[N_c][\mathrm{currentDef}]] \gets \V{o}[\M{L}[N_c][\mathrm{currentDef}]] + 1 $; \\ 
$\V{o}[n_{\mathrm{d}_i}] \gets \V{o}[n_{\mathrm{d}_i}] + 1 $; \\ 
$\V{c}[N_c] \gets \V{c}[N_c] + 1$; \\
$\M{L}[N_c][\V{c}[N_c]] \gets n_{\mathrm{d}_i}$; \\
$\V{p}[n_{\mathrm{d}_i}] \gets N_c$\; 
$\mathrm{totalDef} \gets \mathrm{totalDef} + 1$; \\
$\mathrm{flagEndCluster \gets 0}$\;
}
} 
}
\If{$\mathrm{flagEndCluster = 0}$}
{
$\mathrm{currentDef} \gets \mathrm{currentDef} + 1$\;}
}
}
\end{algorithm} 

%% file: Figures/Algo/Peeling.tex
\begin{algorithm}[t]
\SetKwInOut{Input}{input}
\SetKwInOut{Output}{output}
\caption{Peeling Phase} \label{algo:PP}
\Input{$\V{s}$, $n_\mathrm{d}$, $\M{A}$, $\V{p}$, $\V{c}$, $\M{L}$; \\} 
\Output{$\{ \mathcal{E}_{1}^{1} , \dots, \mathcal{E}_{N_c}^{1} \}$, list of the matchings; \\
$\{ w_1^1 , \dots, w_{N_c}^1 \}$, list of the weights;
}
\BlankLine
\ForAll{$i \in \{1,\dots,N_c\}$}{
$\mathrm{clusterDim}\gets \V{c}[i]$\;
$\mathrm{idx}  \gets \texttt{addGhost}(i)$\;
$\texttt{buildMatch}(\mathrm{idx})$\;
$\V{s}[\mathrm{idx}] \gets 0$\;
\While{$\V{c}[i] > 0$}{
\ForAll{$j \in \{1,\dots,\mathrm{clusterDim} \} $}{
\If{$\V{o}[\M{L}[i][j]] = 1$}{
$\mathrm{idxAdj} \gets \texttt{Adjacent}(\M{L}[i][j])$; \\
\If{$\V{s}[\M{L}[i][j]] = 1$}{
$\texttt{buildMatch}(\M{L}[i][j], \mathrm{idxAdj})$; \\
$\V{s}[\M{L}[i][j]] \gets 0$; \\
$\V{s}[\mathrm{idxAdj}] \gets 1 - \V{s}[\mathrm{idxAdj}]$; \\
} 
\If{$\V{o}[(\mathrm{idxAdj}] = 1$}{
$\V{c}[i] \gets\V{c}[i] - 1 $; \\
}
$\texttt{peel}(\M{L}[i][j],\mathrm{idxAdj})$; \\
$\V{c}[i] \gets \V{c}[i] - 1 $; \\
}
}
}
}
\end{algorithm} 

%% file: Figures/Algo/GhostAdd.tex
\begin{algorithm}[t]
\SetKwInOut{Input}{input}
\SetKwInOut{Output}{output}
\caption{\texttt{addGhost}} \label{algo:ghost}
\Input{$i$, cluster index;} 
\Output{$s_i$, $s_j$, vertices connected to the ghost ancillas;}
\BlankLine
\If{first peeling phase}{
\If{cluster $i$ has an odd number of defects}{
attach a ghost ancilla to the nearest defect $s_i$;\\
}
}
\Else{
\If{cluster $i$ has an odd number of defects}{
attach a ghost ancilla to the nearest defect $s_i$ on the opposite boundary with respect to the boundary selected in peeling phase;\\
}
\Else{
\rev{attach two ghost ancillas: one to the nearest defect $s_i$ on the left boundary, and one to the nearest defect $s_j$ on the right boundary;}
}
}
\end{algorithm}

%% file: Figures/Algo/BuildMatch.tex
\begin{algorithm}[t]
\SetKwInOut{Input}{input}
\SetKwInOut{Output}{output}
\caption{\texttt{buildMatch}} \label{algo:match}
\Input{$s_i$, $s_j$, vertices;\\
$k$, index of the cluster;\\} 
\Output{$\mathcal{E}_k$, current solution (edge/qubit set);}
\BlankLine
Compute the vertical and horizontal coordinates of $s_i$ and $s_j$ using modular arithmetic\;
Start from $s_i$\; 
Move step by step along the vertices in the vertical direction, until reaching a vertex $s_k$ that has the same vertical coordinate as $s_j$\; 
For each step, include in $\mathcal{E}_k$ any edge that connects two consecutively visited vertices\;
Start from $s_k$\;
Move step by step along the vertices in the horizontal direction, until reaching $s_j$\;
For each step, include in $\mathcal{E}_k$ any edge that connects two consecutively visited vertices\;
\end{algorithm}

%% file: Figures/table_cpmplexity.tex
\begin{table*}[t]
    \centering
    \caption{Decoder Complexities.}
    \label{tab:TCompl}
    \begin{tabular}{l c c c c c}
    \toprule
         & MWPM~\cite{Hig:23} & UF~\cite{Del:21} & BC & STM~\cite{ForValChi:24STM}& RFire~\cite{ForValChi:24STM} \\
    \midrule
        complexity &
         $O(n_\mathrm{d}^3 \, \log (n_\mathrm{d}))$ & $O(n\alpha(n))$ & $O(n_\mathrm{d}^2)$ & $O(n_\mathrm{d}^2)\log(n_\mathrm{d})$ & $O(n_\mathrm{d}^2)$ \\
    \bottomrule
    \end{tabular}
\end{table*}